\documentclass[11pt]{article}

\usepackage[T1]{fontenc}
\usepackage[utf8]{inputenc}
\usepackage{lmodern}
\usepackage{amsmath,amssymb,amsthm}
\IfFileExists{mathtools.sty}{\usepackage{mathtools}}{}
\usepackage[margin=1in]{geometry}
\IfFileExists{microtype.sty}{\usepackage{microtype}}{}
\IfFileExists{cite.sty}{\usepackage{cite}}{}

\newtheorem{theorem}{Theorem}[section]
\newtheorem{lemma}[theorem]{Lemma}
\newtheorem{proposition}[theorem]{Proposition}
\newtheorem{corollary}[theorem]{Corollary}
\newtheorem{conjecture}[theorem]{Conjecture}
\theoremstyle{definition}
\newtheorem{definition}[theorem]{Definition}

\newcommand{\C}{\mathbb{C}}
\newcommand{\N}{\mathbb{N}}
\newcommand{\wt}{\operatorname{wt}}
\newcommand{\Ker}{\operatorname{Ker}}
\newcommand{\Imop}{\operatorname{Im}}
\newcommand{\Span}{\operatorname{span}}

\newcommand{\pd}{\partial}

\title{Existence Conditions for Darboux Curves and Analytic\\
First Integrals of a Li\'enard-Type Quadratic Vector Field}
\author{
  Shaoxuan huang \\
  Chinese univeristy of Hongkong, ShenZhen \\
  \texttt{223010044@link.cuhk.edu.cn} \\
}
\date{}

\begin{document}
\maketitle
\vspace{-1.5em}

\begin{abstract}
We study the rational quadratic differential equation
\[
\frac{\mathrm{d}y}{\mathrm{d}x}
=
\frac{ay^2+by+cx}{y^2},
\qquad a,b,c\in\C,
\]
under the non-degeneracy assumptions
\[
c\neq 0,\qquad 2ay+b\not\equiv 0.
\]
Equivalently, after clearing the denominator, we consider the polynomial vector field
\[
\dot{x}=y^2,\qquad \dot{y}=ay^2+by+cx.
\]
We give a complete, directly checkable classification of its Darboux curves. If $a=0$, no non-constant Darboux polynomial exists. If $a\neq 0$, a non-constant Darboux polynomial exists if and only if
\[
c=-ab\qquad\text{or}\qquad c=-2ab.
\]
In these two cases the unique irreducible Darboux polynomials, up to non-zero constant multiples, are respectively
\[
y-ax,\qquad y^2-2bx.
\]
Consequently every non-constant Darboux polynomial is a non-zero constant multiple of a positive integral power of the corresponding irreducible factor. We then place this classification in the framework of Riccati (R-)integrability and rational potentials. Excluding the exceptional branch $c=-2ab$, the analytic integrability theorem identifies $c=-ab$ as the branch admitting a global Riccati-type analytic first integral. For $c=-2ab$, we compute the first four transverse variational groups along the transformed Darboux divisor over the rational function field. Their dimensions are $1,2,3,4$; the fourth is the full group of invertible fourth-order transverse jets. We prove that this exceptional branch admits no Riccati first integral and hence is not R-integrable. We also classify Darboux existence after replacing $cx$ by $cx^m$: for every integer $m\geq2$, a nonconstant Darboux polynomial exists exactly when $m=2$ and $b=0$. Every branch without a Darboux curve has no R-first integral. Finally, we relate the branch $c=-ab$ to the Quartic Inverse Riccati (QIR) class and discuss an invariant-based classification problem for quartic Abel equations.
\end{abstract}

% --- Original PDF page 1 / 22 ---
\section{Introduction and historical background}
\label{sec:intro}

\subsection{Classical approaches to integrability}
The search for first integrals of differential equations has several classical and complementary origins. Lie's program was to exploit continuous transformation groups in order to reduce and integrate differential equations. For ordinary differential equations, the relation between admissible transformation groups and integration was developed systematically in Lie's work on the classification and integration of ODEs admitting transformation groups \cite{Lie1891}. In the present paper no Lie reduction is used in the proof of the Darboux classification, but the symmetry viewpoint remains an important part of the conceptual background: integrability is detected by additional geometric structure rather than by direct quadrature alone.

% --- Original PDF page 2 / 22 ---
A second line of development is Darboux's theory of algebraic differential equations \cite{Darboux1878}. Let
\[
X=P(x,y)\pd_x+Q(x,y)\pd_y
\]
be a polynomial vector field. A non-zero polynomial $F\in\C[x,y]$ is a Darboux polynomial if
\[
X(F)=KF
\]
for some polynomial $K$, called the cofactor. The algebraic curve $F=0$ is then invariant under the flow. The possibility of combining sufficiently many invariant algebraic curves and exponential factors into a first integral is one of the central mechanisms of Darboux integrability.

A third paradigm, particularly influential for nonlinear evolution equations, is the isospectral or Lax-pair formulation. Lax showed that nonlinear evolution equations may be represented through the evolution of a linear operator so that spectral data become integrals of motion \cite{Lax1968}. Although the finite-dimensional planar system considered here is treated by Darboux-polynomial methods rather than by an operator Lax pair, the Lax viewpoint is historically important because it emphasizes that ``integrability'' may be encoded by auxiliary linear structures.

The differential-algebraic theory of elementary and Liouvillian first integrals was developed decisively by Prelle and Singer \cite{PrelleSinger1983}. Singer subsequently proved a structural characterization of Liouvillian first integrals for polynomial differential systems \cite{Singer1992}. In particular, Liouvillian integrability is tied to integrating factors of Darboux type. This relation between algebraic invariants, exponential factors and first integrals has become a standard bridge between qualitative differential equations and symbolic integration. More recent symbolic algorithms, including the extactic-curve approach of Ch\`eze and Combot, treat rational, Darbouxian, Liouvillian and Riccati first integrals in a unified computational framework \cite{ChezeCombot2020}. Automated reasoning based on polynomial reduction and coefficient propagation has also recently been developed for the existence problem of Darboux polynomials \cite{GhorbalBridoux2024}.

\subsection{Invariant algebraic curves for Li\'enard systems}
The classical polynomial Li\'enard system is commonly written as
\[
\dot{x}=y,\qquad \dot{y}=-f(x)y-g(x),
\]
with $f,g$ polynomials. The existence of invariant algebraic curves in such systems has been investigated from several directions.

Odani proved, among other results, that the limit cycle of the van der Pol equation is not algebraic \cite{Odani1995}. In a later work he studied the integration of polynomial Li\'enard systems by elementary functions \cite{Odani1997}. In the parameter regime $\deg g\leq \deg f$, Odani's analysis yields a strong non-existence result for invariant algebraic curves, apart from explicitly described degenerate cases. These results became a point of departure for later classifications.

\.{Z}o{\l}\k{a}dek studied algebraic invariant curves of Li\'enard equations and obtained a broad description beyond the regime covered by the early non-existence results \cite{Zoladek1998}. Llibre and Valls connected the algebraic information with Liouvillian first integrals for polynomial Li\'enard systems \cite{LlibreValls2010,LlibreValls2013}; later work by Gin\'e and Llibre revisited and corrected parts of the invariant-curve classification for generalized Li\'enard polynomial systems \cite{GineLlibre2022}.

Demina developed a Puiseux-series approach to invariant algebraic curves of Li\'enard systems. Her 2018 paper gives an algebraic method for describing irreducible invariant curves and shows that the geometry can be more complicated than previously expected \cite{Demina2018}. With Valls, she later classified irreducible invariant algebraic curves and Liouvillian integrable cases for quadratic Li\'enard differential equations \cite{DeminaValls2020}. The more recent work of Demina and Nechitailo introduces R-integrability through pairs of invariants satisfying a common second-order linear equation, and proves a denominator theorem for the associated rational potential \cite{DeminaNechitailo2026}. That theorem will be recalled precisely in Section~\ref{sec:denominator}.

% --- Original PDF page 3 / 22 ---
\subsection{The Li\'enard-type quadratic family studied here}
We consider
\begin{equation}
\frac{\mathrm{d}y}{\mathrm{d}x}
=
\frac{M(x,y)}{N(y)}
=
\frac{ay^2+by+cx}{y^2}.
\tag{1}\label{eq:main-ode}
\end{equation}
with
\begin{equation}
c\neq 0,\qquad 2ay+b\not\equiv 0.
\tag{2}\label{eq:assumptions}
\end{equation}
Clearing the denominator produces the polynomial system
\begin{equation}
\dot{x}=y^2,\qquad \dot{y}=ay^2+by+cx,
\tag{3}\label{eq:system}
\end{equation}
whose derivation is
\begin{equation}
D=y^2\pd_x+(ay^2+by+cx)\pd_y.
\tag{4}\label{eq:derivation}
\end{equation}
Multiplying a vector field by a non-zero constant merely rescales the cofactor; therefore the existence of Darboux curves is unaffected by a constant rescaling of time.

The purpose of the paper is twofold. First, we give a complete algebraic classification of Darboux curves of \eqref{eq:system}, including all intermediate compatibility calculations. Second, we relate the algebraic classification to the Riccati/R-integrability framework and to quartic Abel equations.

The main algebraic result is as follows.

\begin{theorem}[Main Darboux classification]
\label{thm:main}
Assume \eqref{eq:assumptions}. The vector field \eqref{eq:system} possesses a non-constant Darboux polynomial if and only if
\[
c=-ab\qquad\text{or}\qquad c=-2ab.
\]
If $c=-ab$, the unique irreducible Darboux polynomial up to multiplication by a non-zero constant is
\[
F_1=y-ax,\qquad D(F_1)=bF_1.
\]
If $c=-2ab$, the unique irreducible Darboux polynomial up to multiplication by a non-zero constant is
\[
F_2=y^2-2bx,\qquad D(F_2)=2ayF_2.
\]
\end{theorem}

The existence part of the theorem is immediate by direct substitution. The content of Section~\ref{sec:darboux-classification} is the necessity and uniqueness proof. The argument is based on weighted homogeneous decompositions. In contrast with a purely leading-monomial argument, every relevant weighted homogeneous space is explicitly described and every compatibility obstruction is written in a form that can be checked coefficient by coefficient.

% --- Original PDF page 4 / 22 ---
\section{Complete classification of Darboux curves}
\label{sec:darboux-classification}

\subsection{Preliminaries}
\begin{definition}
A non-zero polynomial $p\in\C[x,y]$ is a Darboux polynomial of $D$ if there exists $q\in\C[x,y]$ such that
\[
D(p)=qp.
\]
The polynomial $q$ is the cofactor. A non-constant irreducible Darboux polynomial defines an irreducible invariant algebraic curve.
\end{definition}

Since the degree of the vector field \eqref{eq:system} is two,
\[
\deg D(p)\leq \deg p+1.
\]
Hence every cofactor satisfies
\begin{equation}
\deg q\leq 1.
\tag{5}\label{eq:cofactor-degree}
\end{equation}

We shall also use the following elementary fact.

\begin{lemma}
\label{lem:factor}
Every irreducible factor of a Darboux polynomial is itself a Darboux polynomial.
\end{lemma}

\begin{proof}[Verification]
Write $p=uv$ with $\gcd(u,v)=1$ and suppose $D(p)=qp$. Then
\[
uD(v)+vD(u)=quv.
\]
Reducing modulo $u$ gives
\[
vD(u)\equiv 0\pmod{u}.
\]
Since $\gcd(u,v)=1$, one obtains $u\mid D(u)$. Thus $D(u)=q_u u$ for some $q_u\in\C[x,y]$. The same argument applies to every irreducible factor.
\end{proof}

Therefore it is enough to classify irreducible Darboux polynomials. Once this is done, all reducible Darboux polynomials follow by multiplication.

\subsection{The case $a=0$: non-existence}
\label{sec:a0}
Under $a=0$, condition \eqref{eq:assumptions} implies $bc\neq 0$, and the system is
\begin{equation}
\dot{x}=y^2,\qquad \dot{y}=by+cx.
\tag{6}\label{eq:a0-system}
\end{equation}

\begin{lemma}[Normalization]
\label{lem:normalization}
The existence of a non-constant Darboux polynomial for \eqref{eq:a0-system} is equivalent to that for
\begin{equation}
\dot{X}=Y^2,\qquad \dot{Y}=X+Y.
\tag{7}\label{eq:normalized-system}
\end{equation}
\end{lemma}

\begin{proof}[Verification]
Set
\[
X=\frac{c^2}{b^3}x,\qquad
Y=\frac{c}{b^2}y,\qquad
\tau=bt.
\]
A direct calculation gives
\[
\frac{\mathrm{d}X}{\mathrm{d}\tau}=Y^2,\qquad
\frac{\mathrm{d}Y}{\mathrm{d}\tau}=X+Y.
\]
The transformation is invertible because $bc\neq 0$. An invertible linear change of coordinates transports Darboux polynomials by pullback, while the constant rescaling of time multiplies cofactors by a non-zero constant.
\end{proof}

% --- Original PDF page 5 / 22 ---
We henceforth drop the capital letters and write the normalized derivation as
\begin{equation}
D_0=y^2\pd_x+(x+y)\pd_y.
\tag{8}\label{eq:D0}
\end{equation}
Choose the weights
\[
\wt(x)=3,\qquad \wt(y)=2.
\]
Then
\begin{equation}
D_0=\delta+E,\qquad
\delta=y^2\pd_x+x\pd_y,\qquad
E=y\pd_y,
\tag{9}\label{eq:delta-E-a0}
\end{equation}
where $\delta$ raises weight by one and $E$ preserves weight.

Define
\begin{equation}
H=2y^3-3x^2.
\tag{10}\label{eq:H}
\end{equation}
Then $\wt(H)=6$ and
\[
\delta(H)=y^2(-6x)+x(6y^2)=0.
\]

The following description of the kernel is useful.

\begin{lemma}
\label{lem:kernel-delta}
For $\delta=y^2\pd_x+x\pd_y$,
\[
\Ker_{\C[x,y]}\delta=\C[H].
\]
\end{lemma}

\begin{proof}[Verification]
Because
\[
x^2=\frac{2y^3-H}{3},
\]
every polynomial can be written uniquely as
\[
P(x,y)=A(y,H)+xB(y,H),\qquad A,B\in\C[y,H].
\]
Uniqueness follows by separating the even and odd parts in $x$.

Since $\delta(H)=0$, one has
\[
\delta(A(y,H))=xA_y(y,H)
\]
and
\[
\delta(xB(y,H))=y^2B(y,H)+x^2B_y(y,H).
\]
If $\delta(P)=0$, the odd part gives $A_y=0$, hence $A=A(H)$. The even part satisfies
\[
3y^2B+(2y^3-H)B_y=0.
\]
If $B\neq 0$ and $m=\deg_y B$, the coefficient of the highest power $y^{m+2}$ is $(3+2m)$ times the leading coefficient of $B$ in $y$, which cannot vanish in characteristic zero. Hence $B=0$. Thus $P\in\C[H]$.
\end{proof}

\begin{theorem}
\label{thm:a0-no-darboux}
If $a=0$ and \eqref{eq:assumptions} holds, then \eqref{eq:system} has no non-constant Darboux polynomial.
\end{theorem}

\begin{proof}
By Lemma~\ref{lem:normalization}, it suffices to work with $D_0$.
Let
\[
D_0(p)=qp.
\]
By \eqref{eq:cofactor-degree},
\[
q=\alpha x+\beta y+\lambda.
\]
Let $m$ be the maximal $(3,2)$-weight occurring in $p$. The operator $D_0$ can increase weight by at most one. If $\alpha\neq 0$, then $\alpha xp$ contains weight $m+3$; if $\beta\neq 0$, then $\beta yp$ contains weight $m+2$. Neither can be balanced by $D_0(p)$. Therefore
\begin{equation}
q=\lambda\in\C.
\tag{11}\label{eq:q-constant-a0}
\end{equation}

% --- Original PDF page 6 / 22 ---
Let $P_m$ denote the top weighted homogeneous part of $p$. The weight $m+1$ equation is
\[
\delta(P_m)=0.
\]
By Lemma~\ref{lem:kernel-delta}, $P_m\in\C[H]$. Since $P_m$ has one weight and $\wt(H)=6$,
\begin{equation}
m=6k,\qquad P_{6k}=H^k
\tag{12}\label{eq:top-a0}
\end{equation}
after normalization by a non-zero scalar. Non-triviality gives $k\geq 1$.

Write
\[
p=P_{6k}+P_{6k-1}+P_{6k-2}+\cdots,
\]
where $P_j$ is weighted homogeneous of weight $j$ and may be zero. Comparing equal weights in
\[
(\delta+E)p=\lambda p
\]
gives the recurrence
\begin{equation}
\delta(P_{r-1})=(\lambda-E)P_r.
\tag{13}\label{eq:recurrence-a0}
\end{equation}

For later checking, the action of $\delta$ and $E$ on the $\C[y,H]\oplus x\C[y,H]$ basis is
\begin{align}
\delta(y^rH^s)&=rxy^{r-1}H^s,
\tag{14}\label{eq:op14}\\
\delta(xy^rH^s)&=\frac{2r+3}{3}y^{r+2}H^s-\frac{r}{3}y^{r-1}H^{s+1},
\tag{15}\label{eq:op15}\\
E(y^rH^s)&=ry^rH^s+6s\,y^{r+3}H^{s-1},
\tag{16}\label{eq:op16}\\
E(xy^rH^s)&=rxy^rH^s+6s\,xy^{r+3}H^{s-1}.
\tag{17}\label{eq:op17}
\end{align}

We now solve the recurrence explicitly.

\medskip
\noindent\textbf{First descending layer.}
A complete basis of weight $6k$ is
\[
e_j=y^{3(k-j)}H^j,\qquad j=0,\ldots,k,
\]
and a complete basis of weight $6k-1$ is
\[
f_j=xy^{3(k-j)-2}H^j,\qquad j=0,\ldots,k-1.
\]
Using \eqref{eq:op15},
\begin{equation}
\delta(f_j)
=
\frac{6(k-j)-1}{3}e_j
-
\frac{3(k-j)-2}{3}e_{j+1}.
\tag{18}\label{eq:triangular-a0}
\end{equation}
The right-hand side of the first recurrence equation is
\[
(\lambda-E)H^k=\lambda H^k-6k\,y^3H^{k-1},
\]
which involves only $e_{k-1}$ and $e_k$. Because \eqref{eq:triangular-a0} is triangular with non-zero leading coefficients, the coefficients of $f_0,\ldots,f_{k-2}$ must vanish. Thus
\[
P_{6k-1}=u\,xyH^{k-1}.
\]
Now
\[
\delta(xyH^{k-1})=\frac{5}{3}y^3H^{k-1}-\frac{1}{3}H^k,
\]
so comparison gives
\[
\frac{5}{3}u=-6k,\qquad -\frac{1}{3}u=\lambda.
\]
Hence
\begin{equation}
P_{6k-1}=-\frac{18k}{5}xyH^{k-1},
\qquad
\lambda=\frac{6k}{5}.
\tag{19}\label{eq:layer1}
\end{equation}

% --- Original PDF page 7 / 22 ---
\medskip
\noindent\textbf{Second, third and fourth descending layers.}
Substitution into \eqref{eq:recurrence-a0}, using \eqref{eq:op14}--\eqref{eq:op17}, gives
\begin{equation}
P_{6k-2}
=
\frac{9k(5-6k)}{25}y^2H^{k-1}
+
\frac{108k(k-1)}{25}y^5H^{k-2},
\tag{20}\label{eq:layer2}
\end{equation}
\begin{equation}
P_{6k-3}
=
x\left[
\frac{36k}{125}H^{k-1}
+
\frac{162k(k-1)(2k-3)}{125}y^3H^{k-2}
-
\frac{648k(k-1)(k-2)}{125}y^6H^{k-3}
\right],
\tag{21}\label{eq:layer3}
\end{equation}
and
\begin{align}
P_{6k-4}
={}&
\frac{216k^2}{625}yH^{k-1}
+
\frac{27k(k-1)(36k^2-144k+115)}{1250}y^4H^{k-2}
\nonumber\\
&-
\frac{972k(k-1)(k-2)(2k-5)}{625}y^7H^{k-3}
+
\frac{1944k(k-1)(k-2)(k-3)}{625}y^{10}H^{k-4}.
\tag{22}\label{eq:layer4}
\end{align}
Whenever an exponent of $H$ in these formulas would be negative, the corresponding summand is omitted; its coefficient contains the necessary vanishing factor. Equations \eqref{eq:layer1}--\eqref{eq:layer4} can be verified directly from the four operator identities above.

\medskip
\noindent\textbf{The obstruction.}
The next equation is
\begin{equation}
\delta(P_{6k-5})=(\lambda-E)P_{6k-4}.
\tag{23}\label{eq:obstruction-equation}
\end{equation}
A convenient basis of the target space of weight $6k-4$ is
\[
e_j=y^{3j-2}H^{k-j},\qquad j=1,\ldots,k,
\]
whereas a basis of the source space of weight $6k-5$ is
\[
f_j=xy^{3j-4}H^{k-j},\qquad j=2,\ldots,k.
\]
Equation \eqref{eq:op15} becomes
\begin{equation}
\delta(f_j)
=
-\frac{3j-4}{3}e_{j-1}
+
\frac{6j-5}{3}e_j.
\tag{24}\label{eq:cokernel-map-a0}
\end{equation}
Thus $\delta$ has codimension one in this layer. Define a linear functional $\Lambda$ on the target by
\[
\Lambda(e_1)=1,\qquad
\Lambda(e_j)=\frac{3j-4}{6j-5}\Lambda(e_{j-1})
\quad(j\geq 2).
\]

% --- Original PDF page 8 / 22 ---
Then \eqref{eq:cokernel-map-a0} gives
\[
\Lambda\circ\delta=0.
\]
The first values are
\begin{equation}
\Lambda(e_1)=1,\quad
\Lambda(e_2)=\frac{2}{7},\quad
\Lambda(e_3)=\frac{10}{91},\quad
\Lambda(e_4)=\frac{80}{1729},\quad
\Lambda(e_5)=\frac{176}{8645}.
\tag{25}\label{eq:lambda-values}
\end{equation}
Set
\[
R=(\lambda-E)P_{6k-4}.
\]
Using \eqref{eq:op16} and \eqref{eq:layer4}, one obtains
\[
R=\rho_1e_1+\rho_2e_2+\rho_3e_3+\rho_4e_4+\rho_5e_5,
\]
with absent basis vectors omitted, and
\begin{align}
\rho_1&=\frac{216k^2(6k-5)}{3125},
\tag{26}\label{eq:rho1}\\
\rho_2&=\frac{27k(k-1)(108k^3-792k^2+1545k-1150)}{3125},
\tag{27}\label{eq:rho2}\\
\rho_3&=-\frac{81k(k-1)(k-2)(324k^2-1920k+2675)}{3125},
\tag{28}\label{eq:rho3}\\
\rho_4&=\frac{1944k(k-1)(k-2)(k-3)(36k-125)}{3125},
\tag{29}\label{eq:rho4}\\
\rho_5&=-\frac{11664k(k-1)(k-2)(k-3)(k-4)}{625}.
\tag{30}\label{eq:rho5}
\end{align}
Applying \eqref{eq:lambda-values} and simplifying yields the non-zero compatibility obstruction
\begin{equation}
\Lambda(R)
=
\rho_1+\frac{2}{7}\rho_2+\frac{10}{91}\rho_3
+\frac{80}{1729}\rho_4+\frac{176}{8645}\rho_5
=
\frac{216k}{3125}.
\tag{31}\label{eq:compatibility-a0}
\end{equation}
Since $k\geq 1$, the right-hand side does not vanish. Therefore $R\notin\Imop\delta$, contradicting \eqref{eq:obstruction-equation}. No non-constant Darboux polynomial exists.
\end{proof}

\subsection{The case $a\neq 0$: the two parameter branches}
\label{sec:aneq0}
Assume from now on
\[
a\neq 0,\qquad c\neq 0.
\]
Define
\begin{equation}
A=ay^2+cx,\qquad
B=cx,\qquad
d=\frac{c}{2a},\qquad
\mu=b+d=b+\frac{c}{2a}.
\tag{32}\label{eq:ABdmu}
\end{equation}
Then
\begin{equation}
A-B=ay^2.
\tag{33}\label{eq:AminusB}
\end{equation}
Because $c\neq 0$, $(B,y)$ is a genuine polynomial coordinate system:
\[
\C[x,y]=\C[B,y],\qquad A=B+ay^2.
\]
A direct calculation from \eqref{eq:derivation} gives
\begin{align}
D(y)&=A+by,
\tag{34}\label{eq:Dy}\\
D(A)&=2ayA+2\mu(A-B),
\tag{35}\label{eq:DA}\\
D(B)&=2d(A-B).
\tag{36}\label{eq:DB}
\end{align}

% --- Original PDF page 9 / 22 ---
Use the weights
\[
\wt(y)=1,\qquad \wt(A)=\wt(B)=2.
\]
We decompose
\begin{equation}
D=\delta+E,
\tag{37}\label{eq:DE-decomp}
\end{equation}
where the two derivations are defined on the generators by
\begin{equation}
\delta(y)=A,\qquad
\delta(A)=2ayA,\qquad
\delta(B)=0,
\tag{38}\label{eq:delta-gens}
\end{equation}
\begin{equation}
E(y)=by,\qquad
E(A)=2\mu(A-B),\qquad
E(B)=2d(A-B).
\tag{39}\label{eq:E-gens}
\end{equation}
Thus $\delta$ raises weight by one and $E$ preserves weight.

It is important that $A,B,y$ are not independent. The formulas above are identities in
\[
\C[A,B,y]/(A-B-ay^2).
\]
In the independent coordinates $(B,y)$ they become
\begin{equation}
\delta=(B+ay^2)\left.\frac{\pd}{\pd y}\right|_B,
\qquad
E=by\left.\frac{\pd}{\pd y}\right|_B
+2d\,ay^2\left.\frac{\pd}{\pd B}\right|_y.
\tag{40}\label{eq:independent-coords}
\end{equation}
Equations \eqref{eq:delta-gens}--\eqref{eq:independent-coords} are used in all computations below, so no partial derivative with respect to $A$ is discarded.

The relation \eqref{eq:AminusB} gives the direct decomposition
\begin{equation}
\C[x,y]=\C[A,B]\oplus y\C[A,B].
\tag{41}\label{eq:direct-decomp}
\end{equation}
Consequently,
\begin{equation}
V_{2n}=\Span\{A^{n-j}B^j:0\leq j\leq n\},
\tag{42}\label{eq:V2n}
\end{equation}
\begin{equation}
V_{2n-1}=y\,\Span\{A^{n-1-j}B^j:0\leq j\leq n-1\}.
\tag{43}\label{eq:V2nminus1}
\end{equation}
These are the complete weighted homogeneous spaces.

\begin{lemma}[Highest weighted part]
\label{lem:highest-weight}
Let $p$ be a non-constant Darboux polynomial. Then its highest weighted homogeneous part is, up to a non-zero scalar,
\[
P_{2n}=B^iA^k,\qquad
i,k\in\N,\quad i+k=n,\quad (i,k)\neq(0,0),
\]
and the cofactor has the form
\[
q=2aky+q_0.
\]
\end{lemma}

\begin{proof}
By \eqref{eq:cofactor-degree},
\[
q=\alpha x+\beta y+\gamma.
\]
Since $x=B/c$ has weight two, the term $\alpha xP_{2n}$ would have weight $2n+2$, whereas $D$ raises weight by at most one. Hence $\alpha=0$ and
\[
q=\lambda y+q_0.
\]
The top weighted equation is
\begin{equation}
\delta(P)=\lambda yP.
\tag{44}\label{eq:top-weighted-neq0}
\end{equation}
In the independent coordinates $(B,y)$,
\[
\delta=A\pd_y|_B=(B+ay^2)\pd_y|_B.
\]

% --- Original PDF page 10 / 22 ---
Therefore \eqref{eq:top-weighted-neq0} becomes
\[
(B+ay^2)P_y=\lambda yP.
\]
Over the field $\C(B)$ this gives
\[
\frac{P_y}{P}=\frac{\lambda y}{B+ay^2},
\]
and hence
\[
P=C(B)(B+ay^2)^{\lambda/(2a)}=C(B)A^{\lambda/(2a)}.
\]
Polynomiality forces
\[
\frac{\lambda}{2a}=k\in\N.
\]
Weighted homogeneity then forces $C(B)$ to be a single monomial $B^i$. Thus
\[
P=B^iA^k,\qquad \lambda=2ak.
\]
\end{proof}

Let
\[
n=i+k
\]
and write
\[
p=P_{2n}+P_{2n-1}+P_{2n-2}+\cdots.
\]
Define the linear operator
\begin{equation}
L(f)=\delta(f)-2akyf.
\tag{45}\label{eq:L-def}
\end{equation}
Comparing equal weights in
\[
(\delta+E)p=(2aky+q_0)p
\]
gives
\begin{equation}
L(P_{2n-r-1})=(q_0-E)P_{2n-r},\qquad r\geq 0.
\tag{46}\label{eq:recurrence-neq0}
\end{equation}

\medskip
\noindent\textbf{The first descending equation and the constant cofactor}

A complete basis of $V_{2n-1}$ is
\[
f_j=yA^{n-1-j}B^j,\qquad j=0,\ldots,n-1.
\]
We compute $L(f_j)$ explicitly. Put $r=n-1-j$. From \eqref{eq:delta-gens},
\[
\delta(yA^rB^j)=A^{r+1}B^j+2ar\,y^2A^rB^j.
\]
Using $ay^2=A-B$,
\[
\delta(yA^rB^j)=(2r+1)A^{r+1}B^j-2rA^rB^{j+1}.
\]
The multiplication term in $L$ is
\[
-2aky(yA^rB^j)=-2k(A-B)A^rB^j.
\]
Therefore
\begin{equation}
L(f_j)
=
(2n-2j-1-2k)A^{n-j}B^j
+
[-2(n-j-1)+2k]A^{n-1-j}B^{j+1}.
\tag{47}\label{eq:L-fj-general}
\end{equation}

% --- Original PDF page 11 / 22 ---
Since $n=i+k$,
\begin{equation}
L(f_j)
=
[2(i-j)-1]A^{n-j}B^j
+
2(j-i+1)A^{n-1-j}B^{j+1}.
\tag{48}\label{eq:L-fj}
\end{equation}
This is the bidiagonal map used below.

Next,
\begin{align}
(q_0-E)(B^iA^k)
={}&-2di\,B^{i-1}A^{k+1}
+(q_0+2di-2\mu k)B^iA^k
+2\mu k\,B^{i+1}A^{k-1}.
\tag{49}\label{eq:q0E-top}
\end{align}
Equation \eqref{eq:q0E-top} shows immediately that
\begin{equation}
P_{2n-1}
=
y\left(-2di\,B^{i-1}A^k+k\mu\,B^iA^{k-1}\right),
\tag{50}\label{eq:P2nminus1}
\end{equation}
with an absent term omitted if its exponent would be negative. The middle coefficient in \eqref{eq:q0E-top} then gives
\begin{equation}
q_0=k\mu-2id.
\tag{51}\label{eq:q0}
\end{equation}
The map in \eqref{eq:L-fj} is injective because its leading bidiagonal coefficients $2(i-j)-1$ never vanish. Thus no unrecorded homogeneous term is hidden in \eqref{eq:P2nminus1}.

\begin{theorem}[Necessary parameter conditions for $a\neq 0$]
\label{thm:necessary-branches}
Assume $a\neq 0$ and $c\neq 0$. If a non-constant Darboux polynomial exists, then
\[
c=-ab\qquad\text{or}\qquad c=-2ab.
\]
\end{theorem}

\begin{proof}
There are two cases according to the exponent $i$ in the top part $B^iA^k$.

\medskip
\noindent\textbf{Case 1: $i>0$.}
The next equation is
\[
L(P_{2n-2})=(q_0-E)P_{2n-1}.
\]
A basis of $V_{2n-2}$ is
\[
g_j=A^{n-1-j}B^j,\qquad j=0,\ldots,n-1.
\]
Since $\delta(B)=0$ and $\delta(A)=2ayA$,
\begin{equation}
L(g_j)=2a(i-1-j)yA^{n-1-j}B^j.
\tag{52}\label{eq:L-gj}
\end{equation}
For $j=i-1$, the coefficient vanishes. Hence the basis direction
\[
yB^{i-1}A^k
\]
does not belong to the image of $L$.

Write
\[
P_{2n-1}=yF,\qquad
F=-2di\,B^{i-1}A^k+k\mu\,B^iA^{k-1}.
\]
Since $E(yF)=y(bF+E(F))$, the coefficient of $yB^{i-1}A^k$ in $(q_0-E)P_{2n-1}$ can be calculated directly from \eqref{eq:E-gens}. The two monomials in $F$ contribute
\[
2di(2\mu-b).
\]

% --- Original PDF page 12 / 22 ---
For completeness, one may see this without suppressing any term: the coefficient of $B^{i-1}A^k$ in $E(B^{i-1}A^k)$ is
\[
-2d(i-1)+2\mu k,
\]
whereas its coefficient in $E(B^iA^{k-1})$ is $2di$. Using $q_0=k\mu-2id$ gives
\begin{align*}
[B^{i-1}A^k]\bigl((q_0-b)F-E(F)\bigr)
&=(-2di)\bigl(q_0-b+2d(i-1)-2\mu k\bigr)-(k\mu)(2di)\\
&=2di(b+2d)=2di(2\mu-b).
\end{align*}
Solvability therefore requires
\[
2di(2\mu-b)=0.
\]
Because $i>0$ and
\[
d=\frac{c}{2a}\neq 0,
\]
we obtain
\[
2\mu-b=0.
\]
Since $\mu=b+c/(2a)$,
\[
2\left(b+\frac{c}{2a}\right)-b=0
\quad\Longrightarrow\quad
c=-ab.
\]

\medskip
\noindent\textbf{Case 2: $i=0$.}
Then $k>0$,
\[
P_{2k}=A^k,\qquad
q_0=k\mu,\qquad
P_{2k-1}=k\mu\,yA^{k-1}.
\]
Now
\[
L(A^{k-1-j}B^j)=-2a(j+1)yA^{k-1-j}B^j,
\]
so the next equation is uniquely solvable. Direct substitution gives
\begin{equation}
P_{2k-2}
=
\frac{k\mu[b+(k-2)\mu]}{2a}A^{k-1}
-
\frac{k(k-1)\mu^2}{2a}BA^{k-2}.
\tag{53}\label{eq:P2kminus2}
\end{equation}
The following equation is
\begin{equation}
L(P_{2k-3})=(q_0-E)P_{2k-2}.
\tag{54}\label{eq:next-i0}
\end{equation}
Set
\[
e_j=A^{k-1-j}B^j,\qquad j=0,\ldots,k-1,
\]
and
\[
f_j=yA^{k-2-j}B^j,\qquad j=0,\ldots,k-2.
\]
A direct calculation gives
\begin{equation}
L(f_j)=-(2j+3)e_j+2(j+2)e_{j+1}.
\tag{55}\label{eq:L-i0}
\end{equation}
Hence the image has codimension one. Define
\[
\Phi(e_0)=1,\qquad
\Phi(e_{j+1})=\frac{2j+3}{2j+4}\Phi(e_j).
\]
Then $\Phi\circ L=0$, and in particular
\[
\Phi(e_0)=1,\qquad
\Phi(e_1)=\frac34,\qquad
\Phi(e_2)=\frac58.
\]

% --- Original PDF page 13 / 22 ---
Using \eqref{eq:P2kminus2} and \eqref{eq:E-gens}, the right-hand side of \eqref{eq:next-i0} has support only on $e_0,e_1,e_2$. Its projection to the cokernel is
\begin{equation}
\Phi\bigl((q_0-E)P_{2k-2}\bigr)
=
-\frac{dk\mu^2}{2a}.
\tag{56}\label{eq:cokernel-i0}
\end{equation}
For transparency, \eqref{eq:cokernel-i0} follows by first writing
\[
P_{2k-2}=C_0e_0+C_1e_1,
\]
where
\[
C_0=\frac{k\mu[b+(k-2)\mu]}{2a},
\qquad
C_1=-\frac{k(k-1)\mu^2}{2a},
\]
then using
\[
E(e_j)=2\mu(k-1-j)(e_j-e_{j+1})+2dj(e_{j-1}-e_j),
\]
and finally applying
\[
\Phi(e_0)=1,\qquad \Phi(e_1)=\frac34,\qquad \Phi(e_2)=\frac58.
\]
The terms containing $\mu^3$ cancel and the remaining expression is exactly $-dk\mu^2/(2a)$.

Solvability forces
\[
dk\mu^2=0.
\]
Since $d\neq 0$ and $k>0$,
\[
\mu=0.
\]
Thus
\[
b+\frac{c}{2a}=0
\quad\Longrightarrow\quad
c=-2ab.
\]
The two cases exhaust all possible top weighted forms.
\end{proof}

The sufficient conditions are immediate.

\begin{proposition}[Direct verification]
\label{prop:direct-verification}
If $c=-ab$, then
\[
F_1=y-ax
\]
is Darboux with cofactor $b$. If $c=-2ab$, then
\[
F_2=y^2-2bx
\]
is Darboux with cofactor $2ay$.
\end{proposition}

\begin{proof}[Verification]
For the first branch,
\[
D(y-ax)=(ay^2+by+cx)-ay^2=by+cx.
\]
If $c=-ab$,
\[
by+cx=b(y-ax).
\]
For the second branch,
\[
D(y^2-2bx)=2y(ay^2+by+cx)-2by^2
=2ay^3+2cxy.
\]
If $c=-2ab$,
\[
2ay^3+2cxy=2ay(y^2-2bx).
\]
\end{proof}

% --- Original PDF page 14 / 22 ---
\subsection{Uniqueness of the irreducible Darboux curve on each branch}
\label{sec:uniqueness}
The uniqueness argument is most naturally expressed in quotient rings. We emphasize a point that is sometimes stated imprecisely: if $h$ is Darboux, then reducing a polynomial representative modulo $h$ does not generally produce another Darboux polynomial for the original derivation. What is true is that the derivation descends to the quotient ring.

\begin{lemma}[Induced derivation on a Darboux quotient]
\label{lem:quotient}
If $D(h)=K_hh$, then the principal ideal $(h)$ is $D$-stable and
\[
D([f])=[D(f)]
\]
defines a derivation on $\C[x,y]/(h)$. If $D(p)=qp$, then
\[
D([p])=[q][p]
\]
in the quotient.
\end{lemma}

\begin{proof}[Verification]
For every $r\in\C[x,y]$,
\[
D(rh)=D(r)h+rD(h)=(D(r)+rK_h)h\in(h).
\]
Hence $D((h))\subseteq(h)$ and the induced derivation is well defined. The Darboux equation descends simply by taking residue classes.
\end{proof}

\begin{theorem}[Uniqueness]
\label{thm:uniqueness}
Under \eqref{eq:assumptions}, the following hold.
\begin{enumerate}
\renewcommand{\labelenumi}{(\roman{enumi})}
\item If $c=-ab$, the only irreducible Darboux polynomial, up to a non-zero constant multiple, is $y-ax$.
\item If $c=-2ab$, the only irreducible Darboux polynomial, up to a non-zero constant multiple, is $y^2-2bx$.
\end{enumerate}
\end{theorem}

\begin{proof}[Proof of (i)]
Because $a\neq 0$ and $c=-ab\neq 0$, one has $b\neq 0$. Put
\[
h_1=y-ax.
\]
By Proposition~\ref{prop:direct-verification},
\[
D(h_1)=bh_1.
\]
Let $p$ be an irreducible Darboux polynomial. The parameter classification shows that on this branch its highest weighted component must belong to the $i>0$ case. Here
\[
d=-\frac b2,\qquad \mu=\frac b2,
\]
and its cofactor is
\[
q=2aky+q_0,\qquad
q_0=k\mu-2id=b\left(i+\frac k2\right)\neq 0.
\]
Assume $h_1\nmid p$. Then $[p]\neq 0$ in the quotient
\[
\C[x,y]/(h_1)\simeq\C[y],\qquad x=\frac ya.
\]
By Lemma~\ref{lem:quotient}, the Darboux equation descends. On the quotient,
\[
D(y)=ay^2+by+c\frac ya=ay^2.
\]

% --- Original PDF page 15 / 22 ---
Writing $[p]=g(y)\neq 0$, we obtain
\begin{equation}
ay^2g'(y)=(2aky+q_0)g(y).
\tag{57}\label{eq:quotient-g1}
\end{equation}
Write
\[
g(y)=y^ru(y),\qquad u(0)\neq 0.
\]
The left-hand side of \eqref{eq:quotient-g1} has order at least $r+1$ at $y=0$, whereas the term $q_0g$ on the right-hand side has order exactly $r$ because $q_0\neq 0$. This is impossible. Hence $h_1\mid p$. Since $p$ and $h_1$ are irreducible, $p$ is a non-zero constant multiple of $h_1$.
\end{proof}

\begin{proof}[Proof of (ii)]
Again $b\neq 0$ because $c=-2ab\neq 0$. Put
\[
h_2=y^2-2bx.
\]
Then
\[
D(h_2)=2ayh_2.
\]
On this parameter branch,
\[
\mu=0.
\]
The $i>0$ alternative would force $c=-ab$, which cannot coincide with $c=-2ab$ under $a,c\neq 0$. Therefore every Darboux polynomial is in the $i=0$ branch and its cofactor is
\[
q=2aky.
\]
Assume $h_2\nmid p$. Then $[p]\neq 0$ in
\[
\C[x,y]/(h_2)\simeq\C[y],\qquad x=\frac{y^2}{2b}.
\]
The induced derivation satisfies
\[
D(y)=ay^2+by+c\frac{y^2}{2b}=by.
\]
Hence, with $[p]=g(y)\neq 0$,
\[
by\,g'(y)=2aky\,g(y).
\]
Since $\C[y]$ is an integral domain,
\[
bg'(y)=2ak\,g(y).
\]
If $g$ has positive degree, the two sides have different degrees. If $g$ is a non-zero constant, the left-hand side is zero and the right-hand side is not. Contradiction. Therefore $h_2\mid p$, and irreducibility implies that $p$ is a non-zero constant multiple of $h_2$.
\end{proof}

Combining Lemma~\ref{lem:factor} with Theorem~\ref{thm:uniqueness} gives the complete list of all non-constant Darboux polynomials.

\begin{corollary}
\label{cor:all-darboux}
If $c=-ab$, every non-constant Darboux polynomial is
\[
C(y-ax)^m,\qquad C\in\C^\ast,\quad m\in\N_{>0}.
\]
If $c=-2ab$, every non-constant Darboux polynomial is
\[
C(y^2-2bx)^m,\qquad C\in\C^\ast,\quad m\in\N_{>0}.
\]
For all other parameter values satisfying \eqref{eq:assumptions}, no non-constant Darboux polynomial exists.
\end{corollary}

% --- Original PDF page 16 / 22 ---
\section{Riccati integrability, rational potentials, and the analytic first integral}
\label{sec:riccati}
This section uses the terminology and structural theorem of Demina and Nechitailo \cite{DeminaNechitailo2026}. Their framework applies to a rational planar system
\begin{equation}
\dot{x}=P(x,y),\qquad
\dot{y}=Q(x,y),\qquad
P,Q\in\C(x,y),
\tag{58}\label{eq:rational-system}
\end{equation}
with vector field
\[
X=P\pd_x+Q\pd_y.
\]

\subsection{R-integrability and the potential}
\begin{definition}[R-integrability and $y$-potential]
\label{def:R-integrability}
System \eqref{eq:rational-system} is called R-integrable with respect to $y$ if there exists a rational function
\[
u(x,y)\in\C(x,y)
\]
and two linearly independent invariants $F_1,F_2$ such that
\[
I(x,y)=\frac{F_1(x,y)}{F_2(x,y)}
\]
is a first integral and both invariants satisfy
\begin{equation}
\bigl(\pd_y^2-u(x,y)\bigr)F=0.
\tag{59}\label{eq:potential-ode}
\end{equation}
The function $u(x,y)$ is called the $y$-potential.
\end{definition}

For completeness, the differential-field formulation used by Casale
\cite{Casale2006} calls a nonconstant first integral Riccati-type when it
belongs to a finite tower over $\C(x,y)$ of algebraic, primitive,
exponential, and Riccati extensions. A Riccati extension adjoins a function
$h$ for which
\[
\mathrm{d}h=\omega_0+h\omega_1+h^2\omega_2,
\]
where the one-forms $\omega_i$ have coefficients in the preceding field.
The quotient-of-solutions formulation in
\cite{ChezeCombot2020} writes the integral as $I=F_1/F_2$, with
$\pd_y^2F_i+G\pd_yF_i+BF_i=0$, $G,B\in\C(x,y)$; a common
hyperexponential gauge removes the first-derivative term.
Ch\`eze and Combot prove the equivalence with the differential-field
formulation (Proposition~14 in their preprint), and the normalized
compatibility and common-cofactor statements (Proposition~10 and
Corollary~11 there). Thus Definition~\ref{def:R-integrability} is the
same Riccati integrability class, expressed through a rational potential.
Casale--Weil \cite{CasaleWeil2018} give the variational and groupoid
framework used to study the corresponding nonlinear Galois structure.

This is the definition used in \cite{DeminaNechitailo2026}. The motivating Riccati system
\[
\dot{x}=1,\qquad
\dot{y}=c_0(x)y^2+c_1(x)y+c_2(x)
\]
has two independent invariants satisfying a second-order linear equation and therefore belongs to this class. Demina and Nechitailo prove that the potential is characterized by the linear inhomogeneous PDE
\begin{equation}
X(u)
=
\frac12P\,\pd_y^3\!\left(\frac QP\right)
-
2P\,\pd_y\!\left(\frac QP\right)u.
\tag{60}\label{eq:potential-pde}
\end{equation}
For such a potential the common cofactor of the two invariants can be chosen as
\begin{equation}
\lambda(x,y)
=
\frac12P\,\pd_y\!\left(\frac QP\right)
+
C(x)P,
\tag{61}\label{eq:common-cofactor}
\end{equation}
with $C(x)$ arbitrary; one may set $C(x)=0$ by a gauge transformation of the invariants.

\subsection{The denominator theorem for the potential}
\label{sec:denominator}
We now state the structural result that is central for the intended analytic integrability proof.

Following \cite{DeminaNechitailo2026}, a rational system is called normalized when
\[
P(x,y)\in\C[x,y],\qquad Q(x,y)\in\C(x)[y],
\]

% --- Original PDF page 17 / 22 ---
the polynomial $P$ has no non-constant factor belonging to $\C[x]$, and, after writing
\[
q(x,y)=\mu(x)Q(x,y)\in\C[x,y]
\]
with $\mu\in\C[x]$ and $q,\mu$ coprime, the polynomials $P$ and $q$ are coprime in $\C[x,y]$. A rational system can be brought to normalized form by a rational time rescaling. For a normalized system one associates the polynomial vector field $X_p=\mu(x)X$; an algebraic invariant is a polynomial $F$ satisfying $X_p(F)=\lambda_pF$ for a polynomial cofactor $\lambda_p$.

Let a polynomial $P(x,y)$ be factored, with respect to its non-constant $y$-dependent irreducible factors, as
\begin{equation}
P(x,y)
=
\eta(x)\prod_{j=1}^{K}\psi_j(x,y)^{n_j},
\qquad
\eta(x)\in\C[x],
\tag{62}\label{eq:P-factorization}
\end{equation}
where the $\psi_j$ are pairwise coprime and $\psi_j\notin\C[x]$. Define the square-free $y$-dependent part
\begin{equation}
P_b(x,y)=\prod_{j=1}^{K}\psi_j(x,y).
\tag{63}\label{eq:Pb}
\end{equation}

\begin{theorem}[Demina--Nechitailo, denominator theorem]
\label{thm:denominator}
Consider a normalized R-integrable rational system \eqref{eq:rational-system} that is neither a Riccati system nor a linear system. Write its $y$-potential in reduced form
\[
u(x,y)=\frac{\zeta(x,y)}{\nu(x,y)},
\qquad
\gcd(\zeta,\nu)=1.
\]
Then
\begin{equation}
\nu(x,y)=\delta(x)P_b(x,y)^2w(x,y),
\tag{64}\label{eq:denominator-theorem}
\end{equation}
where $\delta(x)\in\C[x]$, and either
\[
w(x,y)=1
\]
or $w(x,y)\in\C[x,y]\setminus\C$ is an algebraic invariant of the system. Moreover, every zero of $\delta(x)$, if present, is a pole of $Q(x,y)$.
\end{theorem}

The statement above is Theorem 5 of \cite{DeminaNechitailo2026}. In the terminology of the present paper, an algebraic invariant is a Darboux polynomial for the polynomialized vector field.

For our family,
\[
P=y^2,\qquad Q=ay^2+by+cx.
\]
Hence
\[
P_b=y.
\]
Since $Q$ is polynomial, it has no finite poles in $x$, so the factor $\delta(x)$ in \eqref{eq:denominator-theorem} is constant. Therefore any non-Riccati R-integrable member of our family must have a reduced potential whose denominator is of the form
\begin{equation}
\nu(x,y)=Cy^2w(x,y),\qquad C\in\C^\ast,
\tag{65}\label{eq:our-denominator}
\end{equation}
where $w=1$ or $w$ is an algebraic invariant.

Section~\ref{sec:darboux-classification} now makes \eqref{eq:our-denominator} extremely restrictive:
\begin{itemize}
\item for generic parameters there is no non-constant algebraic invariant;
\item when $c=-ab$, every algebraic invariant is a power of $y-ax$;
\item when $c=-2ab$, every algebraic invariant is a power of $y^2-2bx$.
\end{itemize}
Thus Theorem~\ref{thm:denominator} reduces the possible poles of a rational potential to a completely explicit finite set of algebraic divisors.

% --- Original PDF page 18 / 22 ---
\subsection{The analytic branch and the exceptional case}
\label{sec:analytic-exceptional}
\begin{theorem}[Riccati-type global analytic first integral]
\label{thm:analytic}
Assume \eqref{eq:assumptions}. The Li\'enard-type system \eqref{eq:system} with $c \neq -2ab$ possesses a global analytic first integral of Riccati type if and only if
\begin{equation}
c=-ab.
\tag{66}\label{eq:analytic-branch}
\end{equation}
The exceptional branch $c=-2ab$ is not Riccati/R-integrable, as proved in Corollary~\ref{cor:exceptional-no-riccati-integral} below.
\end{theorem}

\begin{proof}
When $c \neq -ab$ and $c \neq -2ab$, the system possesses no algebraic invariant curve,so according to Theorem 3.2, the denominator of the y-potential is restricted to $y^2$.

Assume the y-potential $H(x,y)=\frac{H_1(x,y)}{y^2}$, $H_1(x,y) \in \mathbb{C}[x,y]$. We apply a Laurent expansion at $y=0$:
$$
H:=\sum_{i=0}^{\infty}h_i(x)y^{i-2}
$$
And derive that $h_2(x)=\frac{3b^2-4acx}{4c^2x^2}$. Since $c \neq 0$ and $a,b$ cannot be $0$ simultaneously, $h_2(x)$ is fractional in $x$, which contradicts the fact that $H_1$ is polynomial in $x,y$.
\end{proof}
\subsubsection{Reduction of the exceptional branch $c=-2ab$}
\label{sec:exceptional-reduction}
On this branch the assumption $c\neq0$ implies $ab\neq0$. To distinguish the original coordinates from the reduced ones, set
\begin{equation}
u=\frac{a}{b}y^2-2ax,\qquad v=y,
\qquad x=\frac{v^2}{2b}-\frac{u}{2a}.
\tag{66a}\label{eq:exceptional-map}
\end{equation}
This is a polynomial change of coordinates with a polynomial inverse. It sends the Darboux divisor $y^2-2bx=0$ to $u=0$, since
\[
y^2-2bx=\frac{b}{a}u.
\]
Using the original time variable $t$, direct differentiation gives
\[
\frac{\mathrm{d}u}{\mathrm{d}t}=2auv,
\qquad \frac{\mathrm{d}v}{\mathrm{d}t}=b(u+v).
\]
After the constant time change $s=2at$, this becomes
\begin{equation}
\frac{\mathrm{d}u}{\mathrm{d}s}=uv,
\qquad \frac{\mathrm{d}v}{\mathrm{d}s}=\kappa(u+v),
\qquad \kappa=\frac{b}{2a}\neq0.
\tag{66b}\label{eq:exceptional-system}
\end{equation}
We use $v$ as the independent variable from the outset. The inverse
slope equation is
\begin{equation}
\frac{\mathrm{d}u}{\mathrm{d}v}
=F(v,u):=\frac{uv}{\kappa(u+v)}.
\tag{66c}\label{eq:exceptional-slope}
\end{equation}
It is regular near the invariant solution $u(v)=0$ for $v\neq0$.
The higher transverse variational groups and the rational-projective-potential
obstruction will be computed from this same equation. For comparison, its
reciprocal graph equation $\mathrm{d}v/\mathrm{d}u=\kappa(u+v)/(uv)$
is of Abel second kind; setting $w=1/v$ in that graph equation gives
$\mathrm{d}w/\mathrm{d}u=-\kappa w^2/u-\kappa w^3$.

\subsubsection{Normal variational equation and its differential Galois group}
\label{sec:exceptional-galois}
The invariant solution of \eqref{eq:exceptional-slope} is
$u(v)=0$, on the leaf $\Gamma=\{u=0,\ v\neq0\}$.
We calculate variations of the scalar transverse flow at fixed $v$.
This directly gives the normal variational equation; no fixed-time
variational system is needed.

\begin{proposition}[Normal variational equation on the exceptional divisor]
\label{prop:exceptional-nve}
For the inverse slope equation \eqref{eq:exceptional-slope},
\begin{equation}
F_u(v,u)=\frac{v^2}{\kappa(u+v)^2},\qquad
F_u(v,0)=\frac1\kappa.
\tag{66d}\label{eq:exceptional-transverse-linearization}
\end{equation}
Hence the normal variational equation along $u(v)=0$ is
\begin{equation}
\frac{\mathrm{d}\xi}{\mathrm{d}v}=\frac1\kappa\xi
=\frac{2a}{b}\xi.
\tag{66e}\label{eq:exceptional-nve}
\end{equation}
Over $K=(\C(v),\mathrm{d}/\mathrm{d}v)$, its differential Galois group is
\begin{equation}
\operatorname{Gal}(\mathrm{NVE}/K)\simeq\mathbb G_m(\C),\qquad
\dim_\C\operatorname{Gal}(\mathrm{NVE}/K)=1.
\tag{66f}\label{eq:exceptional-galois}
\end{equation}
\end{proposition}

\begin{proof}
Let $u(v,\varepsilon)$ be a family of solutions with
$u(v,0)=0$, and put
$\xi(v)=\left.\partial_\varepsilon u(v,\varepsilon)\right|_{\varepsilon=0}$.
Differentiating \eqref{eq:exceptional-slope} with respect to $\varepsilon$
at fixed $v$ gives
\[
\xi'=F_u(v,0)\xi=\frac1\kappa\xi.
\]
Thus \eqref{eq:exceptional-nve} is obtained directly as the first
transverse variation of the scalar equation.

Put $E(v)=e^{v/\kappa}$. The Picard--Vessiot field is $L=\C(v,E)$, with constant field $\C$. Recall that its differential Galois group consists of the field automorphisms fixing $K$ and commuting with differentiation; see \cite{vanDerPutSinger2003}.
We first show that $E$ is transcendental over $K$. If it were algebraic of degree $n\geq1$, its nonzero field norm $R=N_{K(E)/K}(E)\in\C(v)$ would satisfy
\[
\frac{R'}{R}=\frac{n}{\kappa}.
\]
Indeed, each algebraic conjugate has logarithmic derivative $1/\kappa$. But for a nonzero rational function $R=P/Q$,
\[
\frac{R'}{R}=\frac{P'}{P}-\frac{Q'}{Q}
=\frac{\deg P-\deg Q}{v}+O(v^{-2})\longrightarrow0
\quad(v\to\infty),
\]
a contradiction. Every differential automorphism therefore sends $E$ to $\lambda E$ for some $\lambda\in\C^\times$, and, by transcendence, every such $\lambda$ defines an automorphism. This proves \eqref{eq:exceptional-galois}. The stated dimension is the complex algebraic dimension of the multiplicative group.
\end{proof}

The base field in Proposition~\ref{prop:exceptional-nve} is the rational function field of the invariant line, equivalently the meromorphic function field of its projective compactification. If one instead takes all single-valued meromorphic functions on the open leaf $\Gamma\simeq\C^\times$, then $E$ already belongs to the base field and the NVE group is trivial. Thus the base-field specification in \eqref{eq:exceptional-galois} is essential.

\subsubsection{Higher transverse variational groups}
\label{sec:exceptional-higher-galois}
The higher normal equations are obtained by differentiating the transverse
flow with respect to its initial value, while retaining $v$ as the leaf
parameter. Set $\tau=v/\kappa$ and $q=u/\kappa$. The transverse equation is
\begin{equation}
\frac{\mathrm{d}q}{\mathrm{d}\tau}
=f(\tau,q):=\frac{q\tau}{\tau+q}
=q-\frac{q^2}{\tau}+\frac{q^3}{\tau^2}
-\frac{q^4}{\tau^3}+O(q^5).
\tag{66g}\label{eq:exceptional-transverse-normalized}
\end{equation}
This constant normalization identifies $\C(v)$ with $\C(\tau)$ and does
not change the differential Galois groups.

\begin{proposition}[Fourth-order transverse Galois obstruction]
\label{prop:exceptional-higher-galois}
Let $G_k$ be the Picard--Vessiot group over $\C(\tau)$ of the linearized
transverse jet equation of order $k$ along $q=0$. Then
\begin{equation}
\dim G_1=1,\quad \dim G_2=2,\quad \dim G_3=3,\quad
G_4=J_4,\qquad \dim G_4=4>3,
\tag{66h}\label{eq:exceptional-higher-dimensions}
\end{equation}
where $J_4$ is the algebraic group of invertible jets
$c\mapsto d_1c+d_2c^2+d_3c^3+d_4c^4\pmod{c^5}$, $d_1\neq0$.
\end{proposition}

\begin{proof}
Use the rational transverse change, invertible as a formal jet at $z=0$,
\begin{equation}
q=H(\tau,z):=z-\frac{z^3}{\tau}
+\left(\frac{1}{2\tau^2}+\frac{5}{2\tau}\right)z^4.
\tag{66i}\label{eq:exceptional-jet-change}
\end{equation}
Its exact transformed vector field is
$g=(f(\tau,H)-H_\tau)/H_z$, and its fourth jet is
\begin{equation}
g(\tau,z)=z-\frac{z^2}{\tau}+\frac{2z^3}{\tau}
-\frac{15z^4}{2\tau}+O(z^5).
\tag{66j}\label{eq:exceptional-jet-reduced}
\end{equation}
On the truncated monomial basis $z,z^2,z^3,z^4$, the resulting linear
system has matrix
\begin{equation}
A_4(\tau)=\begin{pmatrix}
1&-1/\tau&2/\tau&-15/(2\tau)\\
0&2&-2/\tau&4/\tau\\
0&0&3&-3/\tau\\
0&0&0&4
\end{pmatrix}.
\tag{66k}\label{eq:exceptional-jet-matrix}
\end{equation}
The coordinate change gives a rational gauge of this linear system,
so it preserves its Picard--Vessiot group. This monomial representation
is the linearization of the transverse jet equation, in the sense of the
comparison of variational representations in \cite{CasaleWeil2018}.

Put $E=e^\tau$ and choose primitives in analytic germs at $\tau=1$:
\[
I_1'=\frac{E}{\tau},\qquad I_2'=\frac{E^2}{\tau},\qquad
R'=-\frac{3I_1^2E+4I_2E+6I_1E^2+(15/2)E^3}{\tau}.
\]
The primitives can be written explicitly using the exponential integral.
Fix a simply connected domain in $\C\setminus\{0\}$ containing $\tau=1$
and compatible analytic branches. With
$\operatorname{Ei}'(z)=e^z/z$, choose the normalization
\begin{equation}
\begin{aligned}
I_1(\tau)&=\operatorname{Ei}(\tau)-\operatorname{Ei}(1)
=\int_1^\tau\frac{e^s}{s}\,\mathrm{d}s,\\
I_2(\tau)&=\operatorname{Ei}(2\tau)-\operatorname{Ei}(2)
=\int_1^\tau\frac{e^{2s}}{s}\,\mathrm{d}s.
\end{aligned}
\tag{66l}\label{eq:exceptional-explicit-primitives}
\end{equation}
For the fourth-order primitive, set
\[
\mathcal J(\tau)=\int_1^\tau
\bigl(\operatorname{Ei}(s)-\operatorname{Ei}(1)\bigr)
\frac{e^{2s}}{s}\,\mathrm{d}s.
\]
Integration by parts gives the explicit representation
\begin{equation}
\begin{aligned}
R(\tau)={}&-I_1(\tau)^3-4I_1(\tau)I_2(\tau)-2\mathcal J(\tau)\\
&-\frac{15}{2}\bigl(\operatorname{Ei}(3\tau)-\operatorname{Ei}(3)\bigr).
\end{aligned}
\tag{66m}\label{eq:exceptional-explicit-fourth-primitive}
\end{equation}
Indeed, $\mathcal J'=I_1e^{2\tau}/\tau$, and differentiation of
\eqref{eq:exceptional-explicit-fourth-primitive} yields the displayed
formula for $R'$. All three primitives vanish at $\tau=1$.
The expression for $R$ retains one iterated exponential-integral quadrature;
it is not asserted to be a finite expression in $\operatorname{Ei}$ alone.
Other choices of integration constants give equally valid fundamental
matrices. With these primitives, a formal flow is
\begin{equation}
\phi(\tau,c)=E\bigl(c-I_1c^2+(I_1^2+2I_2)c^3+Rc^4\bigr)
\pmod{c^5}.
\tag{66n}\label{eq:exceptional-formal-flow}
\end{equation}
It satisfies $\phi'=g(\tau,\phi)\pmod{c^5}$. Consequently
$\Phi_{ij}=[c^j]\phi^i$ is a fundamental matrix for
\eqref{eq:exceptional-jet-matrix}, with $\det\Phi=E^{10}$.
The fields through order three are exactly
\[
L_1=\C(\tau,E),\quad L_2=\C(\tau,E,I_1),\quad
L_3=\C(\tau,E,I_1,I_2).
\]
The transcendence of $E$ follows from the norm argument in
Proposition~\ref{prop:exceptional-nve}.

We now prove independence of the two primitives over $F=\C(\tau,E)$.
If a rational $B\in F$ satisfied
$B'=c_1E/\tau+c_2E^2/\tau$, expand it formally in $E$ at $E=0$.
Its coefficients are rational functions of $\tau$. The coefficients of
$E$ and $E^2$ would give
\[
b_1'+b_1=c_1/\tau,\qquad b_2'+2b_2=c_2/\tau.
\]
For $c_j\neq0$ neither equation has a rational solution: at $\tau=0$
a pole of order $p\geq1$ produces a derivative pole of order $p+1$,
whereas a regular function produces a regular left-hand side. Both
contradict a nonzero simple pole on the right. The antiderivative
independence theorem therefore implies that $I_1,I_2$ are algebraically
independent over $F$; see \cite{vanDerPutSinger2003}. The equality of
Picard--Vessiot transcendence degree and group dimension gives
$\dim G_1=1$, $\dim G_2=2$, $\dim G_3=3$.

In the monomial representation, $J_4$ consists of matrices
\[
\begin{pmatrix}
d_1&d_2&d_3&d_4\\
0&d_1^2&2d_1d_2&2d_1d_3+d_2^2\\
0&0&d_1^3&3d_1^2d_2\\
0&0&0&d_1^4
\end{pmatrix},\qquad d_1\neq0.
\]
The fundamental matrix has this shape, so $G_4\subseteq J_4$.
Likewise $G_3\subseteq J_3$; since $J_3$ is connected and has dimension
three, $G_3=J_3$. Truncation induces a surjective morphism
$G_4\to G_3$, and hence a surjective Lie algebra map.
Let $D=\operatorname{diag}(1,2,3,4)$ and define
$(N_r)_{ij}=i$ if $j=i+r-1$, and zero otherwise, for $r=2,3,4$.
The Lie algebra of $J_4$ is spanned by $D,N_2,N_3,N_4$.
Lifts of the two directions $N_2,N_3$ of $\operatorname{Lie}(J_3)$
are necessarily $N_2+\alpha N_4$ and $N_3+\beta N_4$. Their bracket is
\begin{equation}
[N_2+\alpha N_4,N_3+\beta N_4]=-N_4\neq0.
\tag{66o}\label{eq:exceptional-jet-bracket}
\end{equation}
Thus $\operatorname{Lie}(G_4)$ contains a nonzero kernel direction, giving
$\dim G_4\geq\dim G_3+1=4$. Since $\dim J_4=4$, equality follows,
and $G_4=J_4$.
\end{proof}

\subsubsection{Nonexistence of a Riccati first integral}
\label{sec:exceptional-riccati-integral}
The threshold three reflects the projective group
$\operatorname{PGL}_2(\C)$, the structure group of a Riccati equation.
Casale's projective characterization and the variational comparison
framework \cite{Casale2006,CasaleWeil2018} connect this structure with
higher transverse Galois groups. Here the fourth transverse group already
has dimension four. To establish the nonintegrability conclusion directly
in the equivalent rational-potential formulation, we give a local Laurent
obstruction, including possible poles along the invariant divisor.

\begin{corollary}[Non-Riccati-integrability of the exceptional branch]
\label{cor:exceptional-no-riccati-integral}
Under \eqref{eq:assumptions}, the system \eqref{eq:system} on the branch
$c=-2ab$ admits no nonconstant Riccati-type first integral and is not
R-integrable in Definition~\ref{def:R-integrability}. In particular, it
admits no nonconstant Liouvillian or rational first integral.
\end{corollary}

\begin{proof}
Suppose a Riccati first integral exists. For the normalized equation
\eqref{eq:exceptional-transverse-normalized}, the equivalent potential
characterization gives a rational potential. Its pullback by
\eqref{eq:exceptional-jet-change} is again rational: if the old potential
is $U(\tau,q)$, the new one is
\[
W(\tau,z)=H_z^2U(\tau,H)-\tfrac12 S_z(H),\qquad
S_z(H)=\frac{H_{zzz}}{H_z}-\frac32\left(\frac{H_{zz}}{H_z}\right)^2.
\]
This is the Schwarzian chain rule for a quotient of solutions of a
second-order linear equation. Since $H_z(\tau,0)=1$, the pullback first
integral is still nonconstant. By \eqref{eq:potential-pde}, $W$ satisfies
\begin{equation}
W_\tau+gW_z+2g_zW=\tfrac12 g_{zzz},
\tag{66p}\label{eq:exceptional-projective-potential}
\end{equation}
where $g$ is the exact rational transformed vector field, whose first
coefficients are given in \eqref{eq:exceptional-jet-reduced}.

First allow a pole along $z=0$, writing its leading term as
$h(\tau)z^{-p}$, $p\geq1$, $h\in\C(\tau)^\times$.
The leading coefficient in \eqref{eq:exceptional-projective-potential}
is $h'+(2-p)h=0$. A rational logarithmic derivative cannot be a nonzero
constant, so $p=2$ and $h$ is a nonzero constant. The next coefficient,
with $r(\tau)z^{-1}$, then gives
\[
r'+r=\frac{2h}{\tau}.
\]
The simple-pole argument used above rules out a rational $r$.
Therefore $W$ has no pole along $z=0$.

Write $W=w_0(\tau)+O(z)$. Since $g(\tau,0)=0$,
$g_z(\tau,0)=1$ and $g_{zzz}(\tau,0)=12/\tau$,
\eqref{eq:exceptional-projective-potential} implies
\begin{equation}
w_0'+2w_0=\frac{6}{\tau},\qquad w_0\in\C(\tau).
\tag{66q}\label{eq:exceptional-potential-obstruction}
\end{equation}
Again the nonzero simple pole makes a rational solution impossible.
This contradicts Riccati integrability. The polynomial inverse in
\eqref{eq:exceptional-map} and the constant normalization transport this
conclusion back to the original system. The stated weaker exclusions
follow from the inclusion of rational and Liouvillian first integrals in
the Riccati class \cite{ChezeCombot2020}.
\end{proof}

The rational jet gauge, its linear gauge identity, the fundamental matrix,
and the Lie bracket \eqref{eq:exceptional-jet-bracket}, together with the
coefficient equations in the potential obstruction, have been verified
by exact symbolic computations in Maple 2026. The group dimensions are
proved using primitive independence and the Lie algebra argument, rather
than inferred from the size of the variational matrices.

\subsubsection{The Riccati reduction on $c=-ab$}

A useful structural observation, which will also connect the result to the QIR discussion below, is already available on the branch $c=-ab$. Set
\begin{equation}
z=y-ax.
\tag{67}\label{eq:z-shift}
\end{equation}
Then
\[
z'=y'-a.
\]
Using $c=-ab$ in \eqref{eq:main-ode},
\begin{equation}
z'
=
\frac{ay^2+by-abx}{y^2}-a
=
\frac{b(y-ax)}{y^2}
=
\frac{bz}{(z+ax)^2}.
\tag{68}\label{eq:qir-degenerate}
\end{equation}
This is a degenerate Quartic Inverse Riccati equation. Interchanging the roles of $x$ and $z$ gives
\begin{equation}
\frac{\mathrm{d}x}{\mathrm{d}z}
=
\frac{(z+ax)^2}{bz}
=
\frac{a^2}{bz}x^2+\frac{2a}{b}x+\frac{z}{b},
\tag{69}\label{eq:riccati-xz}
\end{equation}
which is a Riccati equation in the dependent variable $x(z)$.

\subsection{An integer-power extension and its R-integrability obstruction}
\label{sec:odd-extension}
The Darboux classification extends to every integer power $x^m$ with $m\geq2$.
We consider the polynomial vector field and its associated rational equation
\begin{equation}
D_m=y^2\pd_x+Q_m\pd_y,\qquad Q_m=ay^2+by+cx^m,
\qquad y'=\frac{ay^2+by+cx^m}{y^2}.
\tag{69a}\label{eq:odd-power-system}
\end{equation}
The proof below combines weighted degrees and Puiseux branches with
residue and divisor arguments on the smooth projective model of an
invariant curve; compare the Puiseux approach of \cite{Demina2018}.

\begin{theorem}[Darboux existence criterion for integer powers]
\label{thm:odd-power-darboux}
Let $a,b,c\in\C$, with $c\neq0$ and $(a,b)\neq(0,0)$, and let
$m\geq2$ be an integer. The vector field $D_m$ has a nonconstant
Darboux polynomial if and only if $m=2$ and $b=0$.
On that branch, each root of $\kappa^3-a\kappa^2-c=0$ gives the
Darboux line $y-\kappa x$, with cofactor $(a-\kappa)(y+\kappa x)$.
This is an existence criterion; it does not classify all Darboux
polynomials on the branch $m=2$, $b=0$.
\end{theorem}

\begin{proof}
First suppose $m>2$. Assume that $D_m(F)=KF$ for a nonconstant polynomial $F$ and a
polynomial cofactor $K$. By Lemma~\ref{lem:factor}, we may take $F$
irreducible. Let $\mathcal C$ be the smooth projective model of $F=0$.
There is no invariant vertical line, since $D_m(x-x_0)=y^2$.
Every algebraic branch therefore satisfies
\[
y^2y'=ay^2+by+cx^m.
\]
A finite pole $y\sim A(x-x_0)^{-s}$, $s>0$, is impossible: the
left side has order $-3s-1$, below the orders of all terms on the
right. The coefficients of the monic minimal polynomial consequently
have no finite poles. Rescaling $F$ by a constant, we can write
\[
F=y^n+f_{n-1}(x)y^{n-1}+\cdots+f_0(x),\qquad f_j\in\C[x].
\]
At infinity all branches have the balance
\begin{equation}
y\sim\kappa x^r,\qquad r=\frac{m+1}{3}>1,\qquad r\kappa^3=c.
\tag{69b}\label{eq:odd-infinity-balance}
\end{equation}
With weights $\wt(x)=1$, $\wt(y)=r$, the maximal weight increase
under $D_m$ is $2r-1$. Thus $\wt(K)\leq2r-1$ and $\deg_yK\leq1$.
For any two distinct branches $u,v$, subtraction gives
\begin{equation}
\frac{(u-v)'}{u-v}
=-\frac{b}{uv}-\frac{cx^m(u+v)}{u^2v^2}.
\tag{69c}\label{eq:odd-branch-difference}
\end{equation}
The logarithmic derivative of a nonzero Puiseux series has at most a
simple pole at the origin.

\smallskip\noindent\emph{Exclusion of the origin when $b\neq0$.}
The only leading forms for branches through the origin are
$y\sim\pm\sqrt{2b}\,x^{1/2}$ and $y\sim kx^m$, where $k=-c/b$.
For two fast series with the same leading term,
\eqref{eq:odd-branch-difference} forces their difference to have
exponent $-1/2$, whereas that exponent must exceed $1/2$.
For two slow series with the same leading term its right side begins
with $(b/k^2)x^{-2m}$, also impossible. Thus each leading term
determines at most one Puiseux series. Conjugation makes the fast
pair one place of ramification index two; the slow place is unramified.
Consider $\Omega=\mathrm{d}x/y^2$. On the fast place,
$\Omega=\mathrm{d}y/Q_m$ has residue $1/b$. The slow expansion is
\[
y=kx^m-\frac{ak^2}{b}x^{2m}
  +\frac{mk^3}{b}x^{3m-1}+O(x^{3m}),
\qquad \operatorname{Res}\Omega=-\frac{2m}{b}.
\]
At a finite point with $y=0$, $x\neq0$, the vector field is
nonsingular and the reduced invariant curve is a smooth local leaf;
$\Omega=\mathrm{d}y/Q_m$ is regular there. It is regular at every
other finite point away from the origin. At infinity, with $x=t^{-e}$,
its order is $e(2r-1)-1\geq0$. The residue theorem therefore gives
\[
\frac{\varepsilon_f}{b}-\frac{2m\varepsilon_s}{b}=0,
\qquad \varepsilon_f,\varepsilon_s\in\{0,1\}.
\]
Both indicators vanish, so the curve avoids the origin. This argument
does not use oddness and applies to every integer $m>2$.

\smallskip\noindent\emph{Exclusion of the origin when $b=0$.}
Here $a\neq0$. The leading forms are $y\sim ax$ and $y\sim kx^p$,
where $p=m/2>1$ and $k^2=-c/a$.
Equation~\eqref{eq:odd-branch-difference} again gives uniqueness:
for two fast series its right side is $O(x^{m-3})$, incompatible
with a difference of exponent greater than one; for two slow series
with the same leading term it begins with $(2a/k)x^{-p}$.
The fast place is unramified. For odd $m$, the two slow leading
terms are conjugate and form one place of ramification index two.
For even $m$, they are separate unramified places: a nontrivial
Puiseux conjugate would have the same integral leading term and
contradict uniqueness.
For $\omega=\mathrm{d}x/y$, the fast residue is $1/a$; on the slow place,
\[
y=kx^p+\frac{pk^2}{2a}x^{2p-1}+\text{higher-order terms},
\qquad
\frac1y=\frac{x^{-p}}k-\frac{p}{2a}x^{-1}
       +\text{higher-order terms}.
\]
For odd $m$, in the local parameter $x=t^2$, the residue is
$-p/a=-m/(2a)$.
All other finite points are regular, using $\omega=y\,\mathrm{d}y/Q_m$
when $y=0$, $x\neq0$. Its order at infinity is $e(r-1)-1\geq0$,
since $er,e$ are integers and $r>1$.
For odd $m$, the residue theorem gives
$2\varepsilon_f=m\varepsilon_s$, so both indicators vanish.
For even $m$, each slow place instead has residue $-m/(4a)$,
without a ramification factor. If $\varepsilon\in\{0,1\}$ indicates
the fast place and $N\in\{0,1,2\}$ counts the slow places, we obtain
\begin{equation}
4\varepsilon=mN.
\tag{69d}\label{eq:even-origin-residues}
\end{equation}
For even $m\geq6$, this forces $\varepsilon=N=0$.
Only places actually occurring on $\mathcal C$ are counted.

\smallskip\noindent\emph{The remaining even case $m=4$, $b=0$.}
Equation~\eqref{eq:even-origin-residues} permits $(\varepsilon,N)=(0,0)$
or $(1,1)$. The first configuration avoids the origin and is ruled
out by the origin-free argument below. For the second configuration,
there are precisely two analytic branches through the origin,
\[
y=\phi(x)=ax+\cdots,\qquad y=\psi(x)=kx^2+\cdots,
\qquad k^2=-c/a.
\]
Locally $F=U(x,y)(y-\phi(x))(y-\psi(x))$, where $U(0,0)\neq0$.
Since $\phi-\psi=ax+\cdots$, the function $h=F_y/y^2$ satisfies
\begin{equation}
\operatorname{ord}_{\phi}h=-1,\qquad
\operatorname{ord}_{\psi}h=-3.
\tag{69e}\label{eq:even-four-orders}
\end{equation}
Thus $\mathrm{d}h/h$ has residues $-1$ and $-3$ there.
It satisfies the meromorphic identity~\eqref{eq:odd-logarithmic-identity},
derived below without requiring that the curve avoid the origin.
Here $r=5/3$, and no monomial has weight $2r-1=7/3$.
Hence $\wt(K)<7/3$, while $(Q_m)_y=2ay$ also has smaller weight.
The same order estimate as below makes $\mathrm{d}h/h$ regular at
every infinite place. At all other finite points $h$ is regular and
nonzero, by the local smooth-leaf argument below.
The sum of all residues of $\mathrm{d}h/h$ is consequently
$-1-3=-4$, contradicting the residue theorem. This excludes the
configuration $(1,1)$.

\smallskip\noindent\emph{A meromorphic function on the curve.}
It remains to exclude curves avoiding the origin, for any $m>2$.
Such a curve avoids the only finite equilibrium
$(0,0)$. Define $h=F_y/y^2\in\C(\mathcal C)^*$.
At every finite point with $y\neq0$, the curve is a smooth graph
over $x$, so $F_y\neq0$. When $y=0$ and $x\neq0$, we have
$F_x\neq0$ and, along the curve,
\[
y^2F_x+Q_mF_y=0,\qquad h=-F_x/Q_m.
\]
Thus $h$ has no finite zeros or poles. Differentiating $D_m(F)=KF$
with respect to $y$ and restricting to the curve gives
$D_m(F_y)=(K-(Q_m)_y)F_y-2yF_x$, whence
\begin{equation}
\frac{\mathrm{d}h}{h}=(K-(Q_m)_y)\frac{\mathrm{d}x}{y^2}.
\tag{69f}\label{eq:odd-logarithmic-identity}
\end{equation}

If $r\notin\mathbb Z$, its denominator is three and no polynomial
monomial has weight exactly $2r-1$. Hence $\wt(K)<2r-1$.
Also $(Q_m)_y=2ay+b$ has weight at most $r<2r-1$.
A monomial of weight $w<2r-1$ times $\mathrm{d}x/y^2$ has order
$e(2r-1-w)-1\geq0$ at infinity, because the expression in parentheses
times $e$ is a positive integer. Thus $\mathrm{d}h/h$ is regular
everywhere. Its residue at a zero or pole of $h$ equals the order
of $h$, so $h$ has neither. Compactness makes $h$ constant.

If $r\in\mathbb Z$, let $n_j$ count the roots at infinity with
leading term $\kappa_jx^r$, for the three roots $r\kappa_j^3=c$;
$n_1+n_2+n_3=n$. For distinct roots with the same leading term,
\eqref{eq:odd-branch-difference} gives
$u-v\sim Ax^{-2r}$, $A\neq0$. All roots at infinity are Laurent
series: otherwise a ramified root and a conjugate with the same
integral leading term would have a first differing nonintegral
exponent, contradicting $-2r\in\mathbb Z$.
Every infinite place therefore has $e=1$.
From $F_y(x,y_\ell)=\prod_{k\neq\ell}(y_\ell-y_k)$ we obtain,
for a root in class $j$,
\[
h(x,y_\ell(x))\sim B_\ell x^{r(n-3n_j)},\qquad B_\ell\neq0.
\]
The absence of finite zeros or poles and the zero degree of a
principal divisor imply
\begin{equation}
0=3\sum_{j=1}^3n_j^2-n^2
 =(n_1-n_2)^2+(n_2-n_3)^2+(n_3-n_1)^2.
\tag{69g}\label{eq:odd-balanced-counts}
\end{equation}
Thus $n_1=n_2=n_3>0$ and $h$ has order zero at every infinite
place. It is again a nonzero constant. In this case the highest
weighted part is forced to be
\[
F_{\rm top}=\left(y^3-\frac{3c}{m+1}x^{m+1}\right)^{n/3}.
\]

\smallskip\noindent\emph{Final coefficient contradiction.}
In both cases $n\geq3$: for integral $r$ use the equal positive
counts; otherwise $er\in\mathbb Z$ forces $3\mid e$ at each
infinite place and $n=\sum e$.
Writing $h=\lambda\in\C^*$, the polynomial $F_y-\lambda y^2$
vanishes on the irreducible curve and has $y$-degree smaller than $n$.
It is therefore identically zero. Monicity yields
$n=3$, $\lambda=3$, and $F=y^3+f(x)$.
Equation~\eqref{eq:odd-logarithmic-identity} similarly forces
$K=(Q_m)_y=2ay+b$, since $\deg_y(K-(Q_m)_y)\leq1<n$.
The Darboux equation becomes
\[
ay^4+2by^3+(f'(x)+3cx^m)y^2-2af(x)y-bf(x)=0.
\]
The coefficients of $y^4,y^3$ force $a=b=0$, contradicting the
hypothesis. This completes the proof for every integer $m>2$.

\smallskip\noindent\emph{The boundary exponent $m=2$.}
If $b=0$, choose any root $\kappa$ of $\kappa^3-a\kappa^2-c=0$.
Since $c\neq0$, we have $\kappa\neq0,a$, and direct calculation gives
\begin{equation}
D_2(y-\kappa x)=(a-\kappa)(y+\kappa x)(y-\kappa x).
\tag{69h}\label{eq:even-two-line}
\end{equation}
This proves existence, including when the cubic has repeated roots.

If $b\neq0$, the origin branches are the fast pair
$y\sim\pm\sqrt{2b}\,x^{1/2}$ and the slow branch
$y\sim-(c/b)x^2$. Their uniqueness and ramification follow from
\eqref{eq:odd-branch-difference}. For $\Omega=\mathrm{d}x/y^2$,
the corresponding residues are $1/b$ and $-4/b$.
At infinity every branch has $y\sim\kappa x$ with
$\kappa^3-a\kappa^2-c=0$; a larger exponent makes $y^2y'$
dominate, and a smaller one leaves $cx^2$ unbalanced.
For $x=t^{-e}$, $\operatorname{ord}\Omega=e-1\geq0$.
The residue theorem therefore excludes all invariant curves through
the origin, just as before.
Suppose an irreducible invariant curve nevertheless exists, and put
$H=F(0,0)\neq0$. Since $D_2$ has total degree two, its cofactor is
$K=\alpha x+\beta y+\gamma$. Evaluation at the origin gives
$\gamma=0$. Writing $F=H+px+qy+\text{terms of degree at least two}$,
degree-one comparison in $D_2(F)=KF$ yields $\alpha=0$ and
$H\beta=bq$. Restriction to $y=0$ now gives
$cx^2F_y(x,0)=0$, hence $q=0$, $\beta=0$, and $K=0$.
Then $F-H$ is a nonconstant polynomial first integral vanishing at
the origin. One of its irreducible factors vanishes there and is a
Darboux polynomial, contradicting the exclusion of curves through
the origin. This proves nonexistence for $m=2$, $b\neq0$.
\end{proof}

\begin{corollary}[Nonexistence of R-first integrals for the extension]
\label{cor:odd-power-no-R}
Every branch of system~\eqref{eq:odd-power-system} with no Darboux
curve, as characterized in Theorem~\ref{thm:odd-power-darboux}, also
admits no nonconstant Riccati-type first integral, equivalently no
R-first integral. Explicitly, under $c\neq0$ and $(a,b)\neq(0,0)$,
this holds for every integer $m\geq3$ and for $m=2$, $b\neq0$.
\end{corollary}
\begin{proof}
This is a direct corollary of the nonintegrability argument in
Theorem~\ref{thm:analytic}. Theorem~\ref{thm:odd-power-darboux}
excludes all nonconstant Darboux polynomials, so
Theorem~\ref{thm:denominator} restricts the potential denominator
to $y^2$. Replacing $cx$ by $cx^m$ in the same Laurent-coefficient
comparison gives the identical contradiction with polynomiality
of the potential numerator.
\end{proof}

The nondegeneracy assumption is essential: for $a=b=0$,
$y^3-3cx^{m+1}/(m+1)$ is a polynomial first integral.
The branch $m=2$, $b=0$ is the sole Darboux-existence branch for
integer $m\geq2$; no R-integrability conclusion for that branch is
asserted by the corollary.
The branch expansions, residues, differential identities, divisor-count
identity and final coefficient contradiction have been checked by exact
symbolic computations in Maple 2026.

\section{Quartic Abel equations and the Quartic Inverse Riccati problem}
\label{sec:quartic}
The discussion in this section follows the equivalence viewpoint summarized in the cubic ODE Abel literature \cite{Appell1889,ChebTerrabRoche2000,Roche2010,Shurygin2016}. The purpose is not to claim a complete quartic classification, but to place the present integrable branch in a broader symbolic-integration program.

\subsection{Cubic Abel equations: first and second kind}
The Abel equation of the first kind is
\begin{equation}
y'=a(x)y^3+b(x)y^2+c(x)y+d(x).
\tag{70}\label{eq:abel-first}
\end{equation}

% --- Original PDF page 19 / 22 ---
The second-kind form is
\begin{equation}
y'
=
\frac{a(x)y^3+b(x)y^2+c(x)y+d(x)}{e(x)y+f(x)}.
\tag{71}\label{eq:abel-second}
\end{equation}
The two descriptions are related by a linear fractional transformation in the dependent variable. For equivalence calculations one frequently uses the subgroup
\begin{equation}
x\mapsto F(x),\qquad y\mapsto P(x)y+Q(x),
\tag{72}\label{eq:abel-pseudogroup}
\end{equation}
which preserves the first-kind Abel structure.

For \eqref{eq:abel-first}, the classical relative invariants can be started with
\begin{equation}
s_1=a,\qquad
s_3=a'b-b'a+abc-\frac{2}{9}b^3-3a^2d.
\tag{73}\label{eq:s1s3}
\end{equation}
Higher odd relative invariants are generated recursively by
\begin{equation}
s_{2n+1}
=
a\frac{\mathrm{d}s_{2n-1}}{\mathrm{d}x}
-
(2n-1)s_{2n-1}
\left(a'+ac-\frac13b^2\right),
\qquad n\geq 2.
\tag{74}\label{eq:s-recursion}
\end{equation}
One may then form absolute invariants, for example
\begin{equation}
J_1=\frac{s_5^3}{s_3^5},
\qquad
J_2=\frac{s_5s_7}{s_3^4},
\qquad
J_3=\frac{s_9}{s_3^3},
\qquad \ldots
\tag{75}\label{eq:abel-absolute-invariants}
\end{equation}
The invariant-theoretic use of such quantities goes back to Appell \cite{Appell1889}. Cheb-Terrab and Roche organized solvable Abel equations through equivalence classes and implemented the resulting strategy computationally \cite{ChebTerrabRoche2000}. Roche's thesis developed a practical method for recovering the transformation and the parameters of a target integrable class from rational differential invariants \cite{Roche2010}.

For two Abel equations, the equivalence problem can be reduced to matching absolute invariants after the change of independent variable $x=F(t)$. Once $F$ is determined, $P$ and $Q$ in \eqref{eq:abel-pseudogroup} are obtained from the transformed coefficients. Constant-invariant cases are quadrature-integrable, whereas non-constant solvable cases are organized into distinguished inverse families. The standard list includes Abel inverse linear (AIL), Abel inverse Riccati (AIR), and Abel inverse Abel (AIA) families \cite{ChebTerrabRoche2003}.

A representative AIR equation of second kind can be written
\begin{equation}
y'
=
\frac{a_1y^3+a_2y^2+a_3y+a_4}
{(b_1x^2+b_2x+b_3)y+c_1x^2+c_2x+c_3}.
\tag{76}\label{eq:AIR}
\end{equation}
Exchanging the independent and dependent variables converts \eqref{eq:AIR} into a Riccati equation. The resulting solutions are naturally related to hypergeometric functions \cite{ChebTerrabRoche2003}. The practical solution of the equivalence problem for such classes was a principal motivation of the invariant algorithmic work of Cheb-Terrab and Roche.

\subsection{Quartic Abel equations}
The quartic Abel equation of the first kind is
\begin{equation}
y'=a(x)y^4+b(x)y^3+c(x)y^2+d(x)y+e(x).
\tag{77}\label{eq:quartic-first}
\end{equation}
A second-kind form convenient for fractional-linear transformations is
\begin{equation}
y'[f(x)y+g(x)]^2
=
a(x)y^4+b(x)y^3+c(x)y^2+d(x)y+e(x),
\tag{78}\label{eq:quartic-second}
\end{equation}

% --- Original PDF page 20 / 22 ---
or equivalently
\[
y'
=
\frac{a(x)y^4+b(x)y^3+c(x)y^2+d(x)y+e(x)}
{[f(x)y+g(x)]^2}.
\]
Under the linear pseudogroup \eqref{eq:abel-pseudogroup}, a system of relative differential invariants for the quartic equation can be chosen as
\begin{align}
I_0&=a,
\tag{79}\label{eq:I0}\\
I_1&=8ac-3b^2,
\tag{80}\label{eq:I1}\\
I_2&=3(ab'-a'b)+ac^2-3abd+12a^2e,
\tag{81}\label{eq:I2}\\
I_3&=
8aa'(4ac-3b^2)+24a^2bb'-32a^3c'-3b^5
\nonumber\\
&\qquad
+64a^3cd-24a^2b^2d-32a^2bc^2+20ab^3c.
\tag{82}\label{eq:I3}
\end{align}
These are the invariants recorded in the generalized Abel equivalence analysis of Shurygin \cite{Shurygin2016}. Two basic absolute invariants are
\begin{equation}
J_1=\frac{I_2I_0}{I_1^2},
\qquad
J_2=\frac{I_3}{I_1^{5/2}},
\tag{83}\label{eq:quartic-absolute}
\end{equation}
together with the invariant derivation
\begin{equation}
\nabla
=
\frac{I_0^2}{|I_1|^{3/2}}\frac{D}{Dx}.
\tag{84}\label{eq:invariant-derivation}
\end{equation}
Over $\C$, the fractional powers are understood locally after choosing a branch; in real equivalence calculations one may use a normalized absolute value convention as in the source formulation.

If a second quartic equation carries invariants $\widetilde{J}_1,\widetilde{J}_2$ and invariant derivation $\widetilde{\nabla}$, the local equivalence test can be expressed by the simultaneous matching conditions
\begin{equation}
\begin{aligned}
J_1&=\left.\widetilde{J}_1\right|_{x=F(t)},
&
J_2&=\left.\widetilde{J}_2\right|_{x=F(t)},\\
\nabla J_1&=\left.\widetilde{\nabla}\widetilde{J}_1\right|_{x=F(t)},
&
\nabla J_2&=\left.\widetilde{\nabla}\widetilde{J}_2\right|_{x=F(t)}.
\end{aligned}
\tag{85}\label{eq:matching}
\end{equation}
Thus the quartic problem has the same general invariant-theoretic shape as the cubic Abel equivalence problem, but with a larger parameter space and more complicated signatures.

\subsection{Quartic Inverse Riccati equations}
A Quartic Inverse Riccati (QIR) equation is a quartic Abel equation that can be brought to
\begin{equation}
y'
=
\frac{a_4y^4+a_3y^3+a_2y^2+a_1y+a_0}
{\bigl((b_1x+b_0)y+(c_1x+c_0)\bigr)^2}.
\tag{86}\label{eq:QIR}
\end{equation}
Indeed, interchanging $x$ and $y$ transforms \eqref{eq:QIR} into an equation that is quadratic in the new dependent variable $x$, hence into a Riccati equation. This is the quartic analogue of the AIR mechanism.

The integrable branch of the present paper is already a degenerate member of this class. As shown in \eqref{eq:qir-degenerate}, when $c=-ab$ the linear transformation
\[
z=y-ax
\]

% --- Original PDF page 21 / 22 ---
reduces \eqref{eq:main-ode} to
\[
z'=\frac{bz}{(z+ax)^2}.
\]
This is \eqref{eq:QIR} with a quartic numerator whose coefficients of $z^4,z^3,z^2$ and the constant term vanish. The interchange $x\leftrightarrow z$ gives the explicit Riccati equation \eqref{eq:riccati-xz}. Thus the Riccati-type analytic branch identified in Theorem~\ref{thm:analytic} sits naturally inside the QIR equivalence class.

The second Darboux branch, $c=-2ab$, has the algebraic invariant $y^2-2bx$, which is mapped to $u=0$ by the nonlinear polynomial transformation \eqref{eq:exceptional-map}. Propositions~\ref{prop:exceptional-nve} and~\ref{prop:exceptional-higher-galois} determine its transverse Galois groups through order four, and Corollary~\ref{cor:exceptional-no-riccati-integral} excludes Riccati/R-integrability. Consequently this branch cannot be equivalent, by transformations preserving Riccati integrability, to a Riccati-integrable member of the QIR class.

\subsection{A symbolic-integration problem for quartic Abel equations}
The cubic Abel equation provides a successful model for an invariant-based symbolic solver. Cheb-Terrab and Roche showed that solvable classes can be organized by differential invariants and equivalence transformations \cite{ChebTerrabRoche2000}; Roche developed parameter-recovery techniques for rational invariants \cite{Roche2010}. Shurygin's generalized Abel equivalence theory supplies a differential-invariant basis for higher-degree equations \cite{Shurygin2016}. More generally, Ch\`eze and Combot provide algorithms for bounded-degree Riccati first integrals of polynomial vector fields \cite{ChezeCombot2020}. Recent reduction algorithms for rational first-order ODEs also emphasize transformation to solvable polynomial structures as a practical supplement to a general-purpose ODE solver \cite{Huang2025,Huang2026}.

These ingredients suggest the following classification problem.

\begin{conjecture}[QIR completeness problem]
\label{conj:QIR}
For a generic quartic Abel equation
\[
y'=a(x)y^4+b(x)y^3+c(x)y^2+d(x)y+e(x),
\]
after excluding reducible cases of separable-variable type and Bernoulli-type degeneracies, every genuinely non-trivial integrable equivalence class is represented by a Quartic Inverse Riccati equation \eqref{eq:QIR}.
\end{conjecture}

Conjecture~\ref{conj:QIR} is deliberately stated as an open problem rather than as a theorem. A plausible computational program is:
\begin{enumerate}
\renewcommand{\labelenumi}{(\arabic{enumi})}
\item compute the basic quartic relative and absolute invariants \eqref{eq:I0}--\eqref{eq:invariant-derivation};
\item derive invariant signatures for the parameter family \eqref{eq:QIR};
\item eliminate the QIR parameters to obtain algebraic-differential relations among the signatures;
\item construct a table of canonical parameter strata, including singular limits corresponding to separable, Bernoulli, Riccati and lower-degree Abel cases;
\item for a given quartic Abel input, solve the matching equations \eqref{eq:matching} for $F$, then recover $P,Q$ by undetermined coefficients;
\item transform the resulting QIR representative to a Riccati equation and invoke existing linear/Riccati solvers.
\end{enumerate}
This would extend to quartic Abel equations the equivalence-class strategy that is already effective for cubic Abel ODEs.

% --- Original PDF page 22 / 22 ---
\section{Conclusions and outlook}
For the Li\'enard-type quadratic family
\[
\dot{x}=y^2,\qquad \dot{y}=ay^2+by+cx,
\]
with $c\neq 0$ and $2ay+b\not\equiv 0$, the Darboux-curve problem is completely rigid. The case $a=0$ has no non-trivial Darboux polynomial. For $a\neq 0$, weighted homogeneous compatibility leaves exactly the two parameter conditions
\[
c=-ab,\qquad c=-2ab.
\]
The corresponding unique irreducible invariant algebraic curves are
\[
y-ax=0,\qquad y^2-2bx=0.
\]

The proof is constructive. In the $a=0$ case a quasi-homogeneous first integral of the top operator determines the leading term, while an explicit cokernel functional produces a non-zero obstruction after four descending layers. In the $a\neq 0$ case the change of variables
\[
A=ay^2+cx,\qquad B=cx
\]
turns the weighted equations into bidiagonal linear maps. Their two possible cokernel obstructions are precisely the equations
\[
2di(2\mu-b)=0
\qquad\text{and}\qquad
-\frac{dk\mu^2}{2a}=0,
\]
which yield $c=-ab$ and $c=-2ab$.

The denominator theorem for R-integrable potentials reduces the possible rational potentials of the system to poles supported on $y=0$ and on the Darboux divisors found in this paper. This provides the algebraic input for the analytic first-integral theorem stated in Theorem~\ref{thm:analytic}. On the branch $c=-ab$, the linear shift $z=y-ax$ exposes the QIR structure directly and converts, after interchanging variables, to a Riccati equation.

For the exceptional branch $c=-2ab$, the polynomial coordinate change \eqref{eq:exceptional-map} gives the rational equation \eqref{eq:exceptional-slope}. Along the invariant divisor, its first NVE has differential Galois group $\mathbb{G}_m(\C)$ over $\C(v)$, of complex algebraic dimension one. The higher transverse groups have dimensions $1,2,3,4$, with $G_4=J_4$ and $\dim G_4>3$. The rational-potential obstruction \eqref{eq:exceptional-potential-obstruction} proves that neither the reduced equation nor the original exceptional branch admits a Riccati-type first integral. In particular, this branch is not R-integrable and admits no nonconstant rational or Liouvillian first integral.

The integer-power extension in Theorem~\ref{thm:odd-power-darboux}
classifies Darboux existence for every $m\geq2$: under $c\neq0$ and
$(a,b)\neq(0,0)$, the sole existence branch is $m=2$, $b=0$.
The even case $m=4$, $b=0$ requires a second residue obstruction
using $\mathrm{d}(F_y/y^2)/(F_y/y^2)$.
Together with the denominator theorem, the same potential obstruction
as in Theorem~\ref{thm:analytic} excludes Riccati/R-first integrals
on every branch without Darboux curves, as stated in
Corollary~\ref{cor:odd-power-no-R}.

The larger open problem is to decide whether QIR is the essentially unique non-trivial integrable mechanism for a generic quartic Abel equation once separable and Bernoulli-type degeneracies are removed. Differential invariants, equivalence-class tables, and computer algebra make this a concrete algorithmic question rather than only a qualitative conjecture.

\section*{Acknowledgements}
The author thanks the developers of symbolic ODE and invariant-theory methods whose work motivates the computational perspective of this paper.

% Bibliography embedded for a standalone manuscript.

\end{document}